\documentclass[journal,twoside,web]{ieeecolor}
\usepackage{lcsys}
\usepackage{cite}
\usepackage{amsmath,amssymb,amsfonts}
\usepackage{algorithmic}
\usepackage{graphicx}
\usepackage{textcomp}
\usepackage{dsfont}
\usepackage{caption}
\usepackage{subcaption}
\newtheorem{assumption}{Assumption}
\newtheorem{proposition}{Proposition}
\newtheorem{theorem}{Theorem}
\newtheorem{remark}{Remark}

\begin{document}
\title{Entropy-Regularized Optimal Transport \\
for Time-Varying Multi-Agent Coverage Control
}
\author{Italo Napolitano, \IEEEmembership{Graduate Student Member, IEEE}, Mario di Bernardo, \IEEEmembership{Fellow, IEEE}
\thanks{(Corresponding author: Mario di Bernardo.) Italo Napolitano and Mario di Bernardo are with the Scuola Superiore Meridionale, Naples, Italy (e-mail: i.napolitano@ssmeridionale.it; mario.dibernardo@unina.it). Mario di Bernardo is also with the Department of Electrical Engineering and Information Technology of the University of Naples Federico II, Naples, Italy.}
}

\maketitle
\thispagestyle{empty}
\begin{abstract}
This paper addresses time-varying coverage control for multi-agent systems, formulated as the tracking of an evolving target density via entropy-regularized semi-discrete optimal transport. Unlike the hard Laguerre partition of the unregularized formulation, entropic regularization assigns fractions of the mass at each point to all agents, simplifying the design and numerical implementation of the control law. We derive a feedback-feedforward controller that tracks the evolving first-order optimality conditions with exponential convergence, and we investigate the role of the regularization parameter, which, above a critical threshold, renders the fully collapsed configuration locally optimal. Numerical experiments validate the theory: for moderate regularization, the proposed approach has performance comparable to the unregularized formulation and outperforms the corresponding Voronoi-based coverage baseline.
\end{abstract}

\begin{IEEEkeywords}
Agents-based systems, Cooperative control, Coverage control, Density control, Entropic regularization, Optimal transport, Swarm robotics, Wasserstein distance
\end{IEEEkeywords}

\section{Introduction}
\label{sec:introduction}
\IEEEPARstart{C}{oordinating} mobile agents toward a desired spatial configuration is a fundamental problem in the control of multi-agent systems. In many applications, however, prescribing individual target positions is unnecessarily restrictive: the objective is instead specified at the collective or probabilistic level through a desired spatial distribution \cite{bandyopadhyay2014probabilistic,lin2025heterogeneous}. This motivates a distributional formulation in which the agents are represented by their empirical measure and are steered toward a possibly time-varying target density, without prescribing an individual target position for each agent.

Optimal transport (OT) provides a natural framework for this problem, as the Wasserstein distance incorporates the geometry of the physical domain and allows discrete agent distributions to be compared directly with the continuous target densities \cite{peyre2019computational}. While several approaches to density control rely on mean-field or continuum descriptions \cite{chen2021optimal,emerick2023continuum,maffettone2025leader}, here we retain the finite-dimensional nature of the multi-agent dynamics and formulate the resulting discrete-to-continuous problem through semi-discrete OT \cite{tacskesen2023semi}. This provides a direct connection between individual agent positions and collective objectives in probability space. This perspective is also related to coverage control, where spatial densities encode the relative importance of different regions and classical formulations lead to Voronoi-based deployments \cite{cortes2004coverage,lee2015multirobot,lin2025heterogeneous}. Recent works have connected coverage objectives with semi-discrete OT and Laguerre partitions \cite{inoue2020optimal,krishnan2022multiscale,krishnan2024distributed}. Building on this interpretation, our previous work \cite{napolitano2026optimal} formulated time-varying coverage as an unregularized semi-discrete OT problem, in which agent positions and Kantorovich dual variables evolve jointly to track the first-order optimality conditions associated with a time-varying density. Such a formulation, however, inherits the hard partition induced by Laguerre cells, so that transport assignments change through moving cell boundaries. This lack of smoothness may be undesirable in a continuously evolving feedback problem, as moving cell boundaries complicate sensitivities, require geometric regularity assumptions, and lead to nonsmooth closed-loop analysis.

In this work, we introduce entropic regularization into the time-varying semi-discrete optimal transport formulation we first proposed in its unregularized form in \cite{napolitano2026optimal}. Entropic OT replaces hard cell-based assignments with smooth transport weights, thereby regularizing the dependence of the transport problem on both the agent positions and the dual variables \cite{peyre2019computational,tacskesen2023semi}. It has also been employed for model predictive control in the discrete-discrete setting~\cite{ito2023entropic}. In contrast to these existing uses, the present work employs entropy not merely as a computational tool, as is common in OT~\cite{cuturi2013sinkhorn}, but as a control-oriented regularization that facilitates the design of smooth continuous-time feedback laws. The regularization parameter provides a natural trade-off between fidelity to the original Wasserstein problem and smoothness of the induced transport allocation.
More broadly, the resulting formulation provides a framework for time-varying multi-agent density tracking and formation control in probability space. Unlike continuum approaches, it acts directly on a finite collection of agents, while allowing the target density to encode time-varying formations without prescribing individual agent assignments. In this Letter, we focus on the centralized setting. Our contributions are threefold: (i) we formulate time-varying density tracking for finite-agent systems as an entropy-regularized semi-discrete OT problem; (ii) we derive the associated regularized optimality conditions and continuous-time dynamics for the agent positions and dual variables, replacing the nonsmooth Laguerre partition with smooth transport weights; and (iii) we investigate the influence of the entropic regularization parameter.

\section{Mathematical Preliminaries}
Scalars are lowercase ($x$), vectors bold lowercase ($\mathbf{x}$), and matrices bold uppercase ($\mathbf{A}$), with maximum eigenvalue $\lambda_\mathrm{max}(\mathbf A)$. The operator $\text{vec}(\cdot)$ stacks vectors, e.g., $\mathbf{p} := \text{vec}(\mathbf{p}_1, \ldots, \mathbf{p}_N)$. The $N$-dimensional identity matrix is denoted by $\mathbf{I}_N$. Matrices of zeros and ones are denoted by $\mathbf{0}_{N \times M}$ and $\mathbf{1}_{N \times M}$, respectively. When $M=1$, they are $\mathbf{0}_N$ and $\mathbf{1}_N$. The Kronecker product is $\otimes$, while $\delta_{ij}$ is the Kronecker delta.
The set of Borel probability measures on $\mathbb R^d$ with finite second moment is $\mathcal{P}_2(\mathbb R^d)$, and $\mathcal{W}_p(\cdot,\cdot)$ denotes the $p$-Wasserstein distance. As in \cite{inoue2020optimal}, $\mathcal{W}_2^2$ is the transport value associated with the half-squared Euclidean cost $\tfrac{1}{2}\lVert\cdot\rVert_2^2$.

\subsection{Semi-Discrete Optimal Transport}
OT compares probability measures by minimizing the cost of mass redistribution. Its Kantorovich formulation yields the Wasserstein distance $\mathcal{W}_p$~\cite{villani2008optimal}. In the semi-discrete setting, continuous and discrete measures are compared and, under mild assumptions, strong duality holds. The optimal map partitions domains into Laguerre cells, which generalize Voronoi cells through additive weights~\cite{aurenhammer1987power}.

Consider a discrete probability measure $\mu := \frac{1}{N}\sum_{i=1}^N \delta_{\mathbf{p}_i}$ with $\mathbf{p}_i\in\mathbb R^d$, and an absolutely continuous measure $\bar\mu \in \mathcal{P}_2(\Omega)$ with density $\bar\rho(\mathbf{x})$ supported in $\Omega \subseteq \mathbb R^d$. For the squared Euclidean cost, the Kantorovich dual formulation~\cite{villani2008optimal,peyre2019computational} reads
\begin{equation*}
\mathcal{W}_2^2(\mu,\bar\mu) = \max_{\boldsymbol{\phi}\in\mathbb{R}^N} F(\mathbf{p},\boldsymbol{\phi}),
\end{equation*}
where $\boldsymbol{\phi}:=\mathrm{vec}(\phi_1,\dots,\phi_N)$ is the dual variable and
\begin{multline}\label{eq:dual_objective}
F(\mathbf{p},\boldsymbol{\phi}) := \sum_{i=1}^N \Bigg[
\int_{\mathcal{L}_i(\mathbf{p},\boldsymbol{\phi})}
\frac{1}{2}\|\mathbf{p}_i-\mathbf{x}\|_2^2\,\bar\rho(\mathbf{x})\,\mathrm{d}\mathbf{x} \\
+ \left(\frac{1}{N} - \int_{\mathcal{L}_i(\mathbf{p},\boldsymbol{\phi})}
\bar\rho(\mathbf{x})\,\mathrm{d}\mathbf{x}\right)\phi_i \Bigg],
\end{multline}
with the $i$-th Laguerre cell (generalized Voronoi)
\begin{multline}\label{eq:laguerre}
\mathcal{L}_i(\mathbf{p},\boldsymbol{\phi}) := \Big\{\mathbf{x}\in\Omega :
\tfrac{1}{2}\|\mathbf{p}_i-\mathbf{x}\|_2^2 - \phi_i \\
\le \tfrac{1}{2}\|\mathbf{p}_j-\mathbf{x}\|_2^2 - \phi_j, \ \forall j \neq i \Big\}.
\end{multline}
\subsection{Entropic Regularization}
Entropy is commonly used as a regularizer in optimization problems over probability measures. In OT, entropic regularization is particularly convenient as it smooths the transport problem and typically reduces the computational burden associated with classical OT \cite{cuturi2013sinkhorn,peyre2019computational,tacskesen2023semi}. In the semi-discrete setting, the regularized problem becomes
\begin{align*}
    \mathcal{W}_{2,\varepsilon}^2(\mu,\bar \mu) = \min_{\boldsymbol \Pi} \sum_{i=1}^N \int_{\Omega} \frac{1}{2} \lVert \mathbf{p}_i - \mathbf{x} \rVert_2^2 \Pi_i(\mathbf x) \bar \rho(\mathbf{x}) \mathrm d\mathbf x \\+ \varepsilon D_\mathrm{KL}(\gamma_{\boldsymbol \Pi} \parallel \bar \mu \otimes \mu) = \max_{\boldsymbol \phi \in \mathbb{R}^N} F^\varepsilon(\mathbf{p}, \boldsymbol \phi), 
\end{align*}
where $D_\mathrm{KL}(\nu_1 \parallel \nu_2)$ is the Kullback-Leibler divergence between measures $\nu_1$ and $\nu_2$, $\boldsymbol \Pi := \mathrm{vec}(\Pi_1(\mathbf{x}), \dots, \Pi_N(\mathbf{x}))$ is the vector of soft Laguerre assignments and $\gamma_{\boldsymbol \Pi}(\mathrm d \mathbf x, i):=\Pi_i(\mathbf x)\bar\rho(\mathbf x)\mathrm d \mathbf x$ is the coupling they induce. The assignments satisfy
\begin{align*}
    \sum_{i=1}^N \Pi_i (\mathbf{x})=1, \; \forall \mathbf x \in \Omega, \\ \int_{\Omega} \Pi_i(\mathbf x) \bar \rho(\mathbf x) \mathrm d\mathbf x= \frac{1}{N}, \; \forall i \in \{1,\dots,N\}.
\end{align*}
The regularized dual objective is
\begin{multline} \label{eq:dual_objective_eps}
F^\varepsilon(\mathbf{p},\boldsymbol\phi):= \frac{1}{N}\sum_{i=1}^N\phi_i \\ -\varepsilon \int_{\Omega} \log\left( \frac{1}{N}\sum_{j=1}^N e^{(\phi_j-c_j(\mathbf{x}))/\varepsilon} \right) \bar\rho(\mathbf{x})\,\mathrm{d}\mathbf{x},
\end{multline}
where the cost of the discrete atom $i$ is $c_i=\tfrac{1}{2} \lVert \mathbf x - \mathbf p_i \rVert_2^2$.
The hard Laguerre assignment is thereby replaced by a softmax distribution: each $\mathbf{x}\in\Omega$ is fractionally assigned to all discrete atoms of $\mu$ rather than belonging to a single Laguerre cell. The $\varepsilon$-regularized assignment associated with atom $i$ is
\begin{equation}
\Pi_i(\mathbf{x}) := \frac{ e^{(\phi_i-c_i(\mathbf{x}))/\varepsilon} }{ \sum_{j=1}^N e^{(\phi_j-c_j(\mathbf{x}))/\varepsilon}}.
\label{eq:soft_laguerre}
\end{equation}
As $\varepsilon\to 0$, the hard Laguerre partition is recovered whenever the limiting Laguerre assignment is unique almost everywhere and its cells have positive mass \cite{peyre2019computational}.
\subsection{Coverage Control}
Coverage control positions mobile agents according to a target density by minimizing a coverage cost~\cite{cortes2004coverage}, classically via Voronoi tessellations~\cite{cortes2004coverage,lee2015multirobot} and, more recently, via semi-discrete OT~\cite{kia2024multi,krishnan2022multiscale,krishnan2024distributed,inoue2020optimal}. Unlike Voronoi partitions, which depend only on agent positions, Laguerre partitions incorporate the dual variables $\boldsymbol{\phi}$ to enforce equal-mass constraints, so that agents concentrate in high-density areas and the finite-agent approximation error is reduced~\cite{napolitano2026optimal,inoue2020optimal}. Building on~\cite{inoue2020optimal}, our previous work~\cite{napolitano2026optimal} tracked the unregularized first-order optimality conditions $\mathbf{p}_i=\mathbf{b}_i$ and $a_i=1/N$, for a time-varying $\bar\rho$, of the optimization problem minimizing the semi-discrete 2-Wasserstein distance squared. That construction, however, requires sensitivities of integrals over moving Laguerre cells: these contain facet integrals and change expression whenever the Laguerre adjacency graph changes, so the analysis needs nondegenerate cells, moving-boundary formulas, and a nonsmooth solution concept, which the entropic framework proposed here avoids.

\section{Modeling and Problem Statement}
\label{sec:problem}
Building upon \cite{napolitano2026optimal}, our aim is to control an ensemble of finite $N$ agents to approximate a continuous, possibly time-varying distribution by minimizing the instantaneous entropic semi-discrete transport cost.

Consider a convex compact domain $\Omega \subseteq \mathbb R^d$, and let $\bar \rho:\Omega \times \mathbb{R}_{\geq0} \to \mathbb{R}_{>0}$ denote a (possibly time-varying) target density function with unitary mass. 
The absolutely continuous target measure $\bar \mu_t \in \mathcal{P}_2(\Omega)$ is induced by $\bar \rho(\mathbf{x},t)$.
\begin{assumption} \label{ass:assumption_ref}
    $\bar \rho$ is $C^1$ and Lipschitz in time.
\end{assumption}
This guarantees that $\bar\rho$ is sufficiently regular and $\partial_t\bar\rho$ exists.
The $i$-th controlled agent is located at position $\mathbf{p}_i(t) \in \mathbb R^d$ at time $t$, and we collect the positions of all $N$ agents in $\mathbf{p} := \mathrm{vec}(\mathbf{p}_1, \dots, \mathbf{p}_N)$.
Each agent evolves according to:
\begin{equation}
\label{eq:agent_model}
\dot{\mathbf{p}}_{i} = \mathbf{u}_i(t), \quad \text{for } i = 1, \dots, N,
\end{equation}
where $\mathbf{u}_i(t) \in \mathbb{R}^d$ is the control input driving the $i$-th agent.
The empirical measure of the agents at time $t$, denoted by $\mu_t$, is defined as the normalized sum of Dirac delta measures centered at the agent positions:
\begin{equation}
\label{eq:empirical_density}
\mu_t := \frac{1}{N} \sum_{i=1}^N \delta_{\mathbf{p}_i}.
\end{equation}

Assuming that the agents have access to $\bar \rho$ and its time derivative $\partial_t \bar \rho$, the objective is to design a control law that computes $\mathbf{u}_i(t)$ for $i = 1, \dots, N$, so that, at each time $t$, the positions of the agents minimize the mismatch between the target measure $\bar\mu_t$ and the empirical measure~\eqref{eq:empirical_density}, i.e., 
\begin{equation}
\label{eq:problem}
\min_{\mathbf{p}} \mathcal{W}^2_{2,\varepsilon}(\mu_t, \bar{\mu}_t)= \min_{\mathbf{p}} \max_{\boldsymbol \phi} F^\varepsilon(\mathbf{p}(t), \boldsymbol \phi(t),t),
\end{equation}
for the entropy regularization parameter $\varepsilon>0$, with $F^\varepsilon$ given by~\eqref{eq:dual_objective_eps}. Note that the \emph{instantaneous} Wasserstein minimization is fundamentally different from dynamic OT in the Benamou-Brenier formulation~\cite{benamou2000computational}.

\section{Entropic Coverage Control} \label{sec:solution}
For every $(\mathbf p,\boldsymbol\phi)$ and every $\varepsilon>0$, the soft assignment~\eqref{eq:soft_laguerre} satisfies $\Pi_i(\mathbf x,\mathbf p(t), \boldsymbol \phi(t))>0$. We define the regularized masses and barycenters as
\begin{subequations} \label{eq:soft_mass_barycenter}
\begin{align}
    a_i^\varepsilon (t) &:= \int_{\Omega} \Pi_i(\mathbf x,\mathbf p(t), \boldsymbol \phi (t)) \bar\rho(\mathbf x, t)\mathrm d \mathbf x, \label{eq:soft_mass}\\
    \mathbf b_i^\varepsilon (t) &:=\frac{1}{a_i^\varepsilon}
    \int_{\Omega} \mathbf x\Pi_i(\mathbf x,\mathbf p(t), \boldsymbol \phi (t))
    \bar\rho(\mathbf x,t)\mathrm d \mathbf x. \label{eq:soft_barycenter}
\end{align}
\end{subequations}
Thus, $a_i^\varepsilon>0$ and every $\mathbf b_i^\varepsilon$ is well defined, without requiring an additional nonempty-cell assumption. Here $a_i^\varepsilon \in \mathbb{R}_{>0}$ and $\mathbf b_i^\varepsilon \in \Omega$, and we collect them as $\mathbf a^\varepsilon := \mathrm{vec}(a_1^\varepsilon,\dots,a_N^\varepsilon) \in \mathbb R^N_{>0}$ and $\mathbf b^\varepsilon := \mathrm{vec}(\mathbf b_1^\varepsilon,\dots,\mathbf b_N^\varepsilon) \in \Omega^N$.
These converge to their unregularized counterparts, the mass and barycenter of the Laguerre cell~\eqref{eq:laguerre}, as $\varepsilon \to 0$~\cite{peyre2019computational}.
Direct differentiation yields
\begin{equation}
    \nabla_{\phi_i}F^\varepsilon=\frac1N-a_i^\varepsilon,
    \quad
    \nabla_{\mathbf p_i}F^\varepsilon
    =a_i^\varepsilon
    (\mathbf p_i-\mathbf b_i^\varepsilon).
    \label{eq:entropic_gradients}
\end{equation}
Setting~\eqref{eq:entropic_gradients} to zero, the regularized first-order conditions are
\begin{subequations} \label{eq:saddle_new}
\begin{align}
    \mathbf p^\star(t) & = \mathbf{b}^\varepsilon(t), \label{eq:min_eps} \\
    \boldsymbol{\phi}^\star(t) & \in \{\boldsymbol{\phi}(t)\mid \mathbf{a}^\varepsilon = \mathbf{1}_N/N\}. \label{eq:max_eps}
\end{align}
\end{subequations}
We will then design dynamics for $(\mathbf p, \boldsymbol \phi)$ to achieve the above conditions. To do so, we first differentiate~\eqref{eq:soft_mass_barycenter} in time, obtaining
\begin{align*}
    \dot{{a}}^\varepsilon_i &= \int_{\Omega} \Pi_i(\mathbf{x}) \partial_t \bar \rho(\mathbf{x},t) \mathrm d\mathbf x + \int_{\Omega} \dot{\Pi}_i(\mathbf{x}) \bar \rho(\mathbf{x},t) \mathrm d\mathbf x, \\
    \dot{\mathbf{b}}^\varepsilon_i &= \frac{1}{a_i^\varepsilon} \int_{\Omega} (\mathbf x-\mathbf{b}_i^\varepsilon)(\Pi_i(\mathbf{x}) \partial_t \bar \rho(\mathbf{x},t) + \dot{\Pi}_i(\mathbf{x}) \bar \rho(\mathbf{x},t) )\mathrm d\mathbf x,
\end{align*}
where $\partial_t \bar \rho$ is given, while 
\begin{align*}
    \dot{\Pi}_i(\mathbf x,\mathbf{p}(t), \boldsymbol \phi(t))= \sum_{j=1}^N \frac{\partial \Pi_i}{\partial \phi_j} \dot \phi_j + \sum_{j=1}^N \frac{\partial \Pi_i}{\partial \mathbf{p}_j} \dot{\mathbf{p}}_j.
\end{align*}
Exploiting the softmax properties, we have
\begin{align*}
    \frac{\partial \Pi_i}{\partial \phi_j} &= \frac{1}{\varepsilon} \Pi_i (\delta_{ij} -\Pi_j), \\
    \frac{\partial \Pi_i}{\partial \mathbf{p}_j} &= \frac{1}{\varepsilon} \Pi_i (\delta_{ij} -\Pi_j) (\mathbf x - \mathbf p_j)^\top.
\end{align*}

Therefore, we define the sensitivity matrix $\mathbf{S}_{pb}^\varepsilon \in \mathbb{R}^{dN \times dN}$, with $d \times d$ block entries
\begin{subequations} \label{eq:sensitivity_matrices}
\begin{align}
    (\mathbf{S}_{pb}^\varepsilon)_{ij} = \frac{1}{\varepsilon a_i^\varepsilon} \int_{\Omega} (\mathbf x- \mathbf b_i^\varepsilon) \Pi_i (\delta_{ij}-\Pi_j)(\mathbf x - \mathbf p_j)^\top  \bar \rho \mathrm d \mathbf x,
\end{align}
the matrix $\mathbf{S}_{\phi b}^\varepsilon \in \mathbb{R}^{dN \times N}$ with $d \times 1$ block entries
\begin{align}
    (\mathbf{S}_{\phi b}^\varepsilon)_{ij} = \frac{1}{\varepsilon a_i^\varepsilon} \int_{\Omega} (\mathbf x- \mathbf b_i^\varepsilon) \Pi_i (\delta_{ij}-\Pi_j)  \bar \rho \mathrm d \mathbf x,
\end{align}
the matrix $\mathbf{S}_{pa}^\varepsilon \in \mathbb{R}^{N \times dN}$ with $1 \times d$ block entries
\begin{align}
    (\mathbf{S}_{pa}^\varepsilon)_{ij} = \frac{1}{\varepsilon} \int_{\Omega} \Pi_i (\delta_{ij}-\Pi_j)(\mathbf x - \mathbf p_j)^\top  \bar \rho \mathrm d \mathbf x,
\end{align}
and the matrix $\mathbf{S}_{\phi a}^\varepsilon \in \mathbb{R}^{N \times N}$ with scalar entries
\begin{align}
    (\mathbf{S}_{\phi a}^\varepsilon)_{ij} = \frac{1}{\varepsilon} \int_{\Omega} \Pi_i (\delta_{ij}-\Pi_j) \bar \rho \mathrm d \mathbf x.
\end{align}
\end{subequations}
This construction allows us to compactly write
\begin{align*}
    \dot{\mathbf{a}}^\varepsilon & = \partial_t \mathbf{a}^\varepsilon + \mathbf{S}_{pa}^\varepsilon \dot{\mathbf{p}} + \mathbf{S}_{\phi a}^\varepsilon \dot{\boldsymbol \phi}, \\
    \dot{\mathbf{b}}^\varepsilon & = \partial_t \mathbf{b}^\varepsilon + \mathbf{S}_{pb}^\varepsilon \dot{\mathbf{p}} + \mathbf{S}_{\phi b}^\varepsilon \dot{\boldsymbol \phi}, 
\end{align*}
with 
\begin{subequations} \label{eq:FF_terms}
    \begin{align} 
    \partial_t {{a}}^\varepsilon_i &= \int_{\Omega} \Pi_i(\mathbf{x}) \partial_t \bar \rho(\mathbf{x},t) \mathrm d\mathbf x, \\
    \partial_t {\mathbf{b}}^\varepsilon_i &= \frac{1}{a_i^\varepsilon} \int_{\Omega} (\mathbf x-\mathbf{b}_i^\varepsilon)\Pi_i(\mathbf{x}) \partial_t \bar \rho(\mathbf{x},t) \mathrm d\mathbf x.
\end{align}
\end{subequations}

We now exploit this structure to design primal and dual dynamics that track the time-varying first-order optimality conditions~\eqref{eq:saddle_new} by directly imposing the desired error dynamics. Since the dual variables $\boldsymbol\phi$ are defined up to an additive constant, we remove one decision variable to resolve this ambiguity and, without loss of generality, set $\phi_N=0$. The inner objective is concave and becomes strictly concave once this gauge is fixed; the outer minimization, however, is nonconvex, as is standard for this class of problems~\cite{inoue2020optimal,napolitano2026optimal,cortes2004coverage}. We therefore target stationary points.

\begin{theorem}[Tracking solution] \label{thm:tracking}
Suppose Assumption \ref{ass:assumption_ref} hold. Let $\mathbf{E}_N:= \begin{bmatrix}
\mathbf{I}_{N-1} & \mathbf{0}_{N-1}
\end{bmatrix}^\top$,
impose $\boldsymbol\phi=\mathbf{E}_N\widehat{\boldsymbol\phi}$ with $\phi_N=0$, and define
\begin{equation}
\label{eq:M}
\mathbf{M}^\varepsilon :=
\begin{bmatrix}
\mathbf{I}_{dN}-\mathbf{S}_{pb}^\varepsilon & -\mathbf{S}_{\phi b}^\varepsilon \mathbf{E}_N \\
\mathbf{E}_N^\top \mathbf{S}_{pa}^\varepsilon & \mathbf{E}_N^\top \mathbf{S}_{\phi a}^\varepsilon \mathbf{E}_N
\end{bmatrix},
\end{equation}
where $\mathbf{S}_{pb}^\varepsilon$, $\mathbf{S}_{\phi b}^\varepsilon$, $\mathbf{S}_{pa}^\varepsilon$, and $\mathbf{S}_{\phi a}^\varepsilon$ are defined in~\eqref{eq:sensitivity_matrices}.
Suppose further that $(\mathbf{M}^\varepsilon)^{-1}$ is well defined\footnote{This assumption is typically made in the related literature on time-varying coverage control, where analogous invertibility conditions arise and are discussed, e.g., in~\cite{lee2015multirobot}}.
Then, the control inputs in~\eqref{eq:agent_model} defined by the continuous-time dynamics
\begin{equation}
\label{eq:tracking_dynamics}
\begin{bmatrix}
\dot{\mathbf p}\\
\dot{\widehat{\boldsymbol\phi}}
\end{bmatrix}
=
(\mathbf{M}^\varepsilon)^{-1}
\begin{bmatrix}
-K_x(\mathbf p-\mathbf b^\varepsilon)+\partial_t\mathbf b^\varepsilon \\
- \mathbf{E}_N^\top \left( K_\phi \left( \mathbf a^\varepsilon - \frac1N\mathbf1_N \right) +\partial_t\mathbf a^\varepsilon \right)
\end{bmatrix},
\end{equation}
with gains $K_x,K_\phi>0$ and feedforward terms $\partial_t\mathbf{a}^\varepsilon$, $\partial_t\mathbf{b}^\varepsilon$ given by~\eqref{eq:FF_terms}, achieve exponential tracking of the first-order optimality conditions~\eqref{eq:saddle_new} of Problem~\eqref{eq:problem}. In particular, the tracking errors $\mathbf{e}_p:=\mathbf{p}-\mathbf{b}^\varepsilon$ and $\mathbf{e}_a:=\mathbf{a}^\varepsilon-\frac{1}{N}\mathbf{1}_N$ converge exponentially to zero with rates $K_x$ and $K_\phi$, respectively.
\end{theorem}

\begin{proof}
    The first block row of~\eqref{eq:tracking_dynamics} gives
    \begin{equation}
    (\mathbf{I}_{dN}-\mathbf{S}_{pb}^\varepsilon)\dot{\mathbf p} - \mathbf{S}_{\phi b}^\varepsilon \mathbf{E}_N\dot{\widehat{\boldsymbol\phi}} = -K_x(\mathbf p-\mathbf b^\varepsilon)+\partial_t\mathbf b^\varepsilon.
\label{pdot}
    \end{equation}
    Since
    \begin{equation}
    \dot{\mathbf b}^\varepsilon = \mathbf{S}_{pb}^\varepsilon \dot{\mathbf p} + \mathbf{S}_{\phi b}^\varepsilon \mathbf{E}_N\dot{\widehat{\boldsymbol\phi}} + \partial_t\mathbf b^\varepsilon,
    \label{bdot}\end{equation}
    subtracting \eqref{bdot} from \eqref{pdot}, we obtain $\dot{\mathbf e}_p = -K_x \mathbf{e}_p$, leading to exponential convergence of the agent positions to \eqref{eq:min_eps} with rate $K_x$.
    The second block row gives
    \begin{align*}
    \mathbf{E}_N^\top \left( \mathbf{S}_{pa}^\varepsilon \dot{\mathbf p} +  \mathbf{S}_{\phi a}^\varepsilon \mathbf{E}_N \dot{\widehat{\boldsymbol\phi}} \right) &= -\mathbf{E}_N^\top \left(K_\phi \mathbf e_a + \partial_t\mathbf a^\varepsilon \right).
    \end{align*}
    Using
    \begin{equation*}
    \dot{\mathbf a}^\varepsilon = \mathbf{S}_{pa}^\varepsilon \dot{\mathbf p} + \mathbf{S}_{\phi a}^\varepsilon \mathbf{E}_N\dot{\widehat{\boldsymbol\phi}} + \partial_t\mathbf a^\varepsilon,
    \end{equation*}
    we obtain
    \begin{equation*}
    \mathbf{E}_N^\top \dot{\mathbf e}_a = -K_\phi \mathbf{E}_N^\top \mathbf e_a.
    \end{equation*}
    Thus, the first $N-1$ components of $\mathbf e_a$ satisfy the desired exponential dynamics. Moreover, since the soft assignments sum to one, $\mathbf1_N^\top\mathbf a^\varepsilon=1$, so that
    \begin{equation*}
    \mathbf1_N^\top\mathbf e_a=0,
    \quad
    \mathbf1_N^\top\dot{\mathbf e}_a=0,
    \end{equation*}
    and the $N$-th component satisfies the same equation. Hence $ \dot{\mathbf e}_a=-K_\phi\mathbf e_a$, leading to exponential convergence of the mass error to zero, thus leading the dual variable to satisfy \eqref{eq:max_eps}.
    Finally, since $\mathbf p=\mathbf b^\varepsilon+\mathbf e_p$ and each $\mathbf b_i^\varepsilon$ is a $\bar\rho$-weighted average of $\mathbf x \in \Omega$, the positions remain bounded, so that the closed loop is forward complete, i.e., no finite escape occurs and its solutions exist for all $t\geq0$.
\end{proof}

The main theoretical advantage over the unregularized formulation of~\cite[Theorem~1]{napolitano2026optimal} is the removal of several geometric regularity assumptions. Since $\Pi_i>0$ for every $\varepsilon>0$, the soft masses, barycenters, and their sensitivities are continuously differentiable at every primal-dual configuration: distinct agents, nondegenerate cells, and locally finite topological transitions of the Laguerre tessellation need no longer be assumed, the closed loop admits classical solutions, and the identities above hold at every time instant rather than almost everywhere. The implementation simplifies accordingly, as all sensitivities are smooth integrals over $\Omega$, evaluated by vectorizable quadrature and softmax operations, with no Laguerre cells to construct, neighbors to identify, or facet integrals to evaluate. The price is the loss of the topology-induced sparsity that enables distributed approximations in the unregularized setting: $\mathbf{M}^\varepsilon$ is generally dense, although approximate sparsity may be recovered by truncating negligible weights.

We now characterize the role of $\varepsilon$. Section~\ref{sec:numerical} shows that small values yield solutions comparable to their unregularized counterparts, so that the above advantages are obtained at a controlled approximation bias. Larger values make the bias more pronounced, observing reduced agent dispersion and, for sufficiently large $\varepsilon$, collapse of the ensemble onto the global barycenter of the target density. Although coincident positions remain admissible in the regularized formulation, such behavior is undesirable when inter-agent collisions must be avoided.
\begin{proposition}[Sub-optimality gap] \label{prop:subopt_gap}
    Let $\mathcal{W}_2^2(\mathbf{p},t)$ and $\mathcal{W}_{2,\varepsilon}^2(\mathbf{p},t)$ be the unregularized and $\varepsilon$-regularized transport costs between empirical and target measures at time $t$ for the given $N$ agents $\mathbf{p}$. Then, for every $(\mathbf{p},t)$,
    \begin{align*}
        0 \leq \mathcal{W}_{2,\varepsilon}^2(\mathbf{p},t)- \mathcal{W}_{2}^2(\mathbf{p},t) \leq \varepsilon \log N.
    \end{align*}
\end{proposition}
\begin{proof}
Since the KL divergence is nonnegative, the value of the regularized problem is no smaller than that of the unregularized one. Conversely, an optimal coupling of the unregularized problem can be chosen deterministic, with \(\Pi_i=\mathds 1_{A_i}\) for measurable sets \(A_i\) satisfying \(\int_{A_i}\bar\rho\mathrm d\mathbf x=1/N\), also when some agents coincide. This coupling is feasible for the regularized problem and $D_{\mathrm{KL}}(\gamma_{\boldsymbol\Pi}\parallel\bar\mu\otimes\mu)=\log N$. Hence,
\begin{align*}
    \mathcal{W}_2^2 \leq \mathcal{W}_{2,\varepsilon}^2 \leq \mathcal{W}_2^2 + \varepsilon \log N.
\end{align*}
\end{proof}
The bound holds uniformly in $\mathbf{p}$ and $t$.
\begin{figure*}[htb]
\centering
    \includegraphics[scale=0.95]{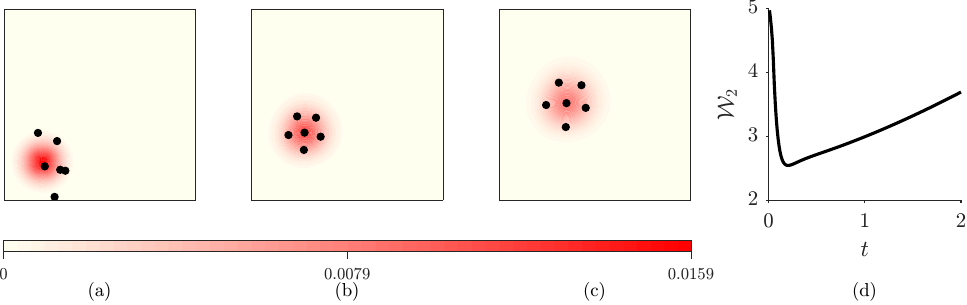}
    \caption{Two-dimensional representative example. Validation of the control law with $6$ agents (black dots) covering a time-varying density (background colour from yellow to red, detailed in the colorbar) at (a) $t=0$, (b) $t=1$, and (c) $t=2$. (d) Time evolution of the 2-Wasserstein distance between the empirical and target measure.}
    \label{fig:exp2D}
\end{figure*}

\begin{proposition}[Collapsed stationary configuration]  \label{prop:collapsed}
    Define the moments of the target distribution
    \begin{align*}
        \bar{\mathbf x} &:= \int_{\Omega} \mathbf{x} \bar \rho (\mathbf x,t) \mathrm d \mathbf x, \\
        \boldsymbol{\Sigma} & := \int_{\Omega} (\mathbf{x}-\bar{\mathbf x}) (\mathbf{x}-\bar{\mathbf x})^\top  \bar \rho (\mathbf x,t) \mathrm d \mathbf x.
    \end{align*}
    For every $\varepsilon>0$, the configuration 
    \begin{align} \label{eq:collapsed_point}
        \mathbf{p}_1=\dots=\mathbf p_N=\bar{\mathbf x}, \quad \boldsymbol \phi = \mathbf{0}_N
    \end{align}
    is a stationary point with Hessian
    \begin{align}\label{eq:hessian_collapsed}
        \nabla_{\mathbf{pp}}^2 \mathcal{W}_{2,\varepsilon}^2 = \frac{1}{N} \left[\mathbf{I}_{dN} - \frac{1}{\varepsilon} (\mathbf{P}_N \otimes \boldsymbol{\Sigma})\right],
    \end{align}
    where
    \begin{align}
        \mathbf{P}_N := \mathbf{I}_N - \frac{1}{N} \mathbf{1}_N \mathbf{1}_N^\top. \label{eq:PN}
    \end{align}
    Therefore, the collapsed point \eqref{eq:collapsed_point} is a strict local minimum if $\varepsilon>\lambda_\mathrm{max}(\boldsymbol{\Sigma})$, and a saddle point if $\varepsilon<\lambda_\mathrm{max}(\boldsymbol{\Sigma})$. 
\end{proposition}
\begin{proof}
    At the collapsed point \eqref{eq:collapsed_point}, 
    \begin{align*}
        \Pi_i = \frac{1}{N}, \quad a_i^\varepsilon=\frac{1}{N}, \quad \mathbf{b}_i^\varepsilon = \bar{\mathbf{x}},
    \end{align*}
    therefore, it is a stationary point. 
    We then compute the Hessian
    \begin{align*}
        \nabla_\mathbf{p} \mathcal{W}_{2,\varepsilon}^2 = \frac{1}{N}(\mathbf p - \mathbf b^\varepsilon), \quad \nabla_{\mathbf{pp}}^2 \mathcal{W}_{2,\varepsilon}^2 = \frac{1}{N} \left(\mathbf{I}_{dN} - \frac{\mathrm d \mathbf{b}^\varepsilon}{\mathrm d \mathbf{p}}\right).
    \end{align*}
    Substituting \eqref{eq:collapsed_point} into the sensitivity matrices concerning $\mathbf{b}^\varepsilon(\mathbf{p},\boldsymbol{\phi}^\star(\mathbf p))$ yields
    \begin{align*}
        \mathbf{S}_{\phi b}^\varepsilon = \mathbf 0_{dN \times N}, \quad \mathbf{S}_{pb}^\varepsilon = \frac{1}{\varepsilon} (\mathbf{P}_N \otimes \boldsymbol{\Sigma}),
    \end{align*}
    for $\mathbf{P}_N$ defined in \eqref{eq:PN}. Since $\mathbf{S}_{\phi b}^\varepsilon$ vanishes, the total derivative $\mathrm d \mathbf b^\varepsilon/\mathrm d \mathbf p = \mathbf{S}_{pb}^\varepsilon + \mathbf{S}_{\phi b}^\varepsilon \mathbf{E}_N \,\mathrm d\widehat{\boldsymbol\phi}^\star/\mathrm d\mathbf p=\mathbf{S}_{pb}^\varepsilon$, and substitution into the Hessian leads to \eqref{eq:hessian_collapsed}. The matrix $\mathbf{P}_N$ has eigenvalue zero along the consensus direction and eigenvalue one along the $N-1$ disagreement directions. Hence, the Hessian has eigenvalue $1/N$ with multiplicity $d$, along the consensus directions, and eigenvalues
    $\frac1N\left(1-\frac{\lambda_k(\boldsymbol{\Sigma})}{\varepsilon}\right)$,
    each with multiplicity $N-1$, along the disagreement directions associated with the $k$th eigenvalue of $\boldsymbol{\Sigma}$, proving the classification.
\end{proof}

\begin{remark}
    By Proposition~\ref{prop:collapsed}, $\varepsilon$ should be chosen strictly below $\lambda_\mathrm{max}(\boldsymbol{\Sigma}(t))$ for all $t$. The threshold, however, governs only \emph{full} collapse: partial clustering of subsets of agents can occur below it, as in Fig.~\ref{fig:exp1D_eps6}. As in the unregularized case, collision-free trajectories are never guaranteed. Characterizing such partial collapses, and enforcing collision constraints, is left for future work.
\end{remark}

\section{Numerical Validation} \label{sec:numerical}
In all experiments, the integrals in~\eqref{eq:sensitivity_matrices} are evaluated by tensor-product quadrature on a uniform grid of $50$ points per dimension in the 2D experiment and $5000$ points in the 1D experiment, \eqref{eq:tracking_dynamics} is integrated with forward Euler with $\Delta t=0.01$, the gains are $K_x=10$, $K_\phi=10$, and $\boldsymbol{\phi}(0)=\mathbf{0}_N$.
We first validate the proposed approach in a representative two-dimensional example. We consider $N=6$ agents tasked with covering an evolving density and set $\varepsilon=6$. The initial density is Gaussian $\bar \rho(\mathbf x,0) = \mathcal{N}([10, 10]^\top, 10 \mathbf{I}_2)$ and evolves according to the advection-diffusion equation
\begin{align} \label{eq:exp2D_FP}
    \partial_t \bar \rho (\mathbf x,t) + \nabla \cdot (\mathbf{v}\bar \rho(\mathbf x,t)) = D \Delta \bar \rho(\mathbf x,t),
\end{align}
complemented with no-flux boundary conditions.
We consider a uniform drift $\mathbf v=[4,8]^\top$ and a diffusion coefficient $D=2$.
This setting may represent, for instance, a contaminant plume transported by a uniform fluid current while diffusing~ \cite{zahugi2013oil,lin2025heterogeneous}. The agents must therefore continuously redistribute to monitor the evolving concentration field.
Figure~\ref{fig:exp2D}a--c shows snapshots at the initial, intermediate, and final times, with \(\partial_t\bar\rho\) computed directly from~\eqref{eq:exp2D_FP}. Starting from random initial positions uniform in $[0,15]^2$, the agents redistribute to cover the density while tracking its drift, and progressively spread out as it diffuses, adapting to the varying target density. The mass and positions error also exponentially decay to zero (their evolutions are omitted here for the sake of brevity). The Wasserstein distance decreases rapidly during the initial transient and then grows linearly, as diffusion increases the variance of the target density and hence the finite-agent quantization error (Fig. \ref{fig:exp2D}d).

Next, we investigate the effect of $\varepsilon$ in a one-dimensional domain. We consider $N=5$ agents covering a Gaussian density with a time-varying mean, namely, $\bar\rho(x,t)=\mathcal{N}(m(t),9)$, where $\dot m(t)=2t$ and $m(0)=10$, over a time horizon $t \in [0,5]$. Figures~\ref{fig:exp1D_eps01}--\ref{fig:exp1D_eps10} show the final agent configurations for different values of $\varepsilon$. For $\varepsilon=0.1$, the agents effectively cover the density while maintaining distinct positions: three agents are closely spaced near the peak of the Gaussian, whereas the remaining two lie farther away toward its tails. For the intermediate value $\varepsilon=6$, the three central agents collapse to the same position, while the two outer agents remain distinct. Finally, for $\varepsilon=10>\lambda_\mathrm{max}(\Sigma)=9$, all agents collapse to a single point corresponding to the global barycenter of the target density, consistently with Proposition~\ref{prop:collapsed}.
Figure~\ref{fig:exp1D:metrics} further confirms these theoretical findings, where we vary $\varepsilon \in [0.1, 15]$. For small values of $\varepsilon$, the standard deviation of the agent positions, normalized by its value in the unregularized case, remains close to one, indicating that the ensemble spreads similarly to the baseline configuration. As $\varepsilon$ increases, the normalized standard deviation decreases and becomes zero near the theoretical critical threshold, remaining zero for all values exceeding the variance of the target density. Accordingly, the normalized Wasserstein distance remains close to one for small $\varepsilon$ and increases to approximately 3 when all agents collapse.

\begin{figure*}[htb]
    \centering
    \vspace{0.3cm}
    \subfloat[]
    {
    \includegraphics{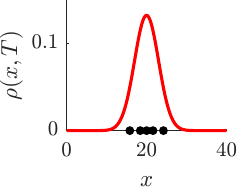}
    \label{fig:exp1D_eps01}
    }
    \subfloat[]
    {
    \includegraphics{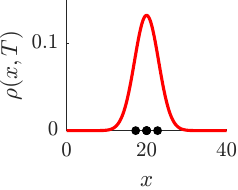}
    \label{fig:exp1D_eps6}
    }
    \subfloat[]
    {
    \includegraphics{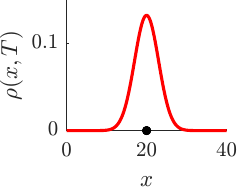}
    \label{fig:exp1D_eps10}
    }
    \subfloat[]
    {
    \includegraphics{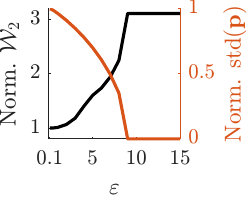}
    \label{fig:exp1D:metrics}
    }
    \caption{Validation of the control law with five agents (black dots) covering the target density (red solid line) at the final time for (a) \(\varepsilon=0.1\), (b) \(\varepsilon=6\), and (c) \(\varepsilon=10\). (d) Effect of \(\varepsilon\) on the final approximation error, measured in terms of \(\mathcal{W}_{2}\) (black line), and on the ensemble spread, measured by \(\operatorname{std}(\mathbf p)\) (orange line). Both quantities are normalized by their corresponding values for the unregularized solution.}
    \label{fig:exp1D}
\end{figure*}

Finally, we compare the proposed solution with the Voronoi-based approach~\cite{lee2015multirobot}, obtained in the unregularized case by setting \(\boldsymbol{\phi}(t)=\mathbf{0}_N\) for all \(t\geq 0\). Since the Voronoi partition does not enforce equal cell masses, the agents concentrate less around the Gaussian peak than under the unregularized OT solution, resulting in a larger final approximation error. Moreover, the entropic controller remains more accurate for all tested values \(\varepsilon\leq 3\), as reported in Table~\ref{tab:comparison} for four representative values of \(\varepsilon\). Finally, the suboptimality gap \(\mathcal{W}_{2,\varepsilon}^2-\mathcal{W}_2^2\) consistently remains below the upper bound \(\varepsilon\log N\) established in Proposition~\ref{prop:subopt_gap}. Notice that at a collapsed configuration, all agents coincide and the regularized optimal coupling has zero KL divergence, consequently, \(\mathcal{W}_{2,\varepsilon}^2=\mathcal{W}_2^2\).
\begin{table}[htb]
    \centering
    \begin{tabular}{|c|c|c|c|c|c|c|c|}
    \hline
         & Voronoi & $\varepsilon=0$ & $\varepsilon=0.1$ & $\varepsilon=3$ & $\varepsilon=4$ & $\varepsilon=15$ \\
         \hline
        $\mathcal{W}_2^2$ & $0.76$ & $0.46$ & $0.46$ & $0.65$ & $0.91$ & $4.5$ \\
        $\mathcal{W}_{2,\varepsilon}^2 - \mathcal{W}_2^2$ & -- & -- & $0.16$ & $2.52$ & $    2.72$ & $0.00$ \\
        $\varepsilon \log  N$ & -- & -- & $0.16$ & $4.83$ & $6.44$ & $24.14$ \\
        \hline
    \end{tabular}
    \caption{Quantitative metric comparison between Voronoi-based solution \cite{lee2015multirobot}, unregularized OT ($\varepsilon=0$) \cite{napolitano2026optimal}, and proposed solution for representative values of $\varepsilon \in [0.1, 15]$.}
    \label{tab:comparison}
\end{table}

\section{Conclusion}
This paper introduced an entropy-regularized optimal transport framework for time-varying coverage control, with application to multi-agent density tracking. The smooth assignment induced by entropic regularization enabled a control law that exponentially tracks the evolving first-order optimality conditions, together with an explicit threshold above which the collapsed configuration becomes locally optimal. Numerical experiments confirmed the theory and, for moderate regularization, have comparable performance with our unregularized approach~\cite{napolitano2026optimal}, while outperforming Voronoi-based coverage~\cite{lee2015multirobot}. Future work will focus on distributed implementations, adaptive selection of $\varepsilon$, extensions to constrained dynamics and collision avoidance, and a theoretical characterization of the conditions ensuring the invertibility of $\mathbf M^\varepsilon$.

\bibliographystyle{IEEEtran}        
\bibliography{refs}

\end{document}